\documentclass[12pt]{article}

\usepackage{a4wide}

\usepackage{enumerate}
\usepackage{amsfonts}
\usepackage{amssymb}
\usepackage{amsthm}
\usepackage{epsfig} 
\usepackage{psfrag}

\begin{document}

\newcommand{\ie}{\emph{i.e.}}
\newcommand{\suchthat}{such that}
\newcommand{\trivpoint}{\item$\bullet$~}
\newcommand{\itemif}{\item(\emph{if}).~}
\newcommand{\itemonlyif}{\item(\emph{only if}).~}

\newcommand{\pmast}{\textup{MAST}}
\newcommand{\past}{\textup{AST}}
\newcommand{\pastD}{\textup{AST}$[D]$}
\newcommand{\pct}{\textup{CT}}
\newcommand{\pctD}{\textup{CT}$[2^{\lfloor D \quotient 2 \rfloor}]$}
\newcommand{\pctE}{\textup{CT}$[D]$}
\newcommand{\pmct}{\textup{MCT}}

\newcommand{\ppisqk}{\textup{PIS}$_{p + 1}[k]$}
\newcommand{\ppisp}{\textup{PIS}$_p$}
\newcommand{\ppispk}{\textup{PIS}$_p[k]$}
\newcommand{\ppis}{\textup{PIS}$_1$}
\newcommand{\ppisk}{\textup{PIS}$_1[k]$}
\newcommand{\pis}{\textup{IS}}
\newcommand{\pisk}{\textup{IS}$[k]$}
\newcommand{\ppisdeux}{\textup{PIS}$_2$}
\newcommand{\ppisdeuxk}{\textup{PIS}$_2[k]$}

\newcommand{\N}{\mathbb{N}}
\newcommand{\seg}[2]{[#1, #2]}
\newcommand{\restr}[2]{#1  {\restriction} #2}
\newcommand{\setnabot}{-}
\newcommand{\ceil}[1]{\left\lceil #1 \right\rceil}
\newcommand{\quotient}{\mathbin{/}}

\newenvironment{proofqed}{\begin{proof}}{\end{proof}}

\newcommand{\defdecpb}[3]{
\begin{trivlist}
\item\hspace{\parindent}{\bf Name:} #1
\item\hspace{\parindent}{\bf Instance:} #2
\item\hspace{\parindent}{\bf Question:} #3 
\end{trivlist}
}

\newcommand{\calc}{\mathcal{C}}
\newcommand{\calt}{\mathcal{T}}

\newcommand{\AST}{agreement subtree}

\newcommand{\lr}[1]{\left\langle #1 \right\rangle}
\newcommand{\feuille}[1]{L(#1)}

\renewcommand{\tilde}{\widetilde}

\newcommand{\WI}{\mathrm{W}[1]}
\newcommand{\NP}{\mathrm{NP}}
\newcommand{\SNP}{\mathrm{SNP}}
\newcommand{\SE}{\mathrm{SE}}
\newcommand{\lfptred}{linearly FPT-reduce}

\newtheorem{lemma}{Lemma}
\newtheorem{theorem}{Theorem}
\newtheorem{definition}{Definition}
\newtheorem{remark}{Remark}
\newtheorem{property}{Property}


\sloppy 
\title{Solving the Maximum Agreement SubTree and the Maximum Compatible Tree problems on many bounded degree trees\thanks{The paper is a revised version of the conference paper \cite{GuillemotN06}.}}

\author{Sylvain~Guillemot  \and Fran\c{c}ois~Nicolas\thanks{Corresponding author. E-mail address: \texttt{nicolas@cs.helsinki.fi}.}}



\maketitle

\begin{abstract}
Given a set of leaf-labeled trees with identical leaf sets, 
the well-known \textsc{Maximum Agreement SubTree} problem  (\pmast) consists of finding a subtree homeomorphically included in all input trees and with the largest number of leaves.
Its variant called \textsc{Maximum Compatible Tree} (\pmct) is less stringent, 
as it allows  the input trees to be refined. 
Both problems are of particular interest in computational biology,
where trees encountered have often small degrees.

In this paper,
 we study the parameterized complexity of \pmast{} and \pmct{} with respect to the maximum degree, 
 denoted by $D$,
 of the input trees.
Although \pmast{} is polynomial for bounded $D$ \cite{AK97,Fal95,B97}, 
we show that the problem is $\WI$-hard with respect to parameter $D$.
Moreover, 
relying on recent advances in parameterized complexity we obtain a tight lower bound:
while \pmast{} can be solved in $O(N^{O(D)})$ time where $N$ denotes the input length,
we show that an $O(N^{o(D)})$ bound is not achievable, unless $\SNP \subseteq \SE$.
We also show that \pmct{} is $\WI$-hard with respect to $D$,  
and that \pmct{} cannot be solved in  
$O\big(N^{o(2^{D \quotient 2})}\big)$ time,
unless $\SNP \subseteq \SE$.  
\end{abstract} 

\section{Introduction}

Throughout this paper, 
$\N$ denotes the set of non-negative integers and, 
for all $n \in \N$, the set $\{ 1, 2, \ldots, n \}$ is denoted by $\seg{1}{n}$.

\subsection{Agreement subtree and compatible tree} \label{sec:agr-comp}

\subsubsection{Trees}

All trees considered in this paper are \emph{rooted evolutionary trees}, 
\ie{} trees representing the evolutionary history of a set of species.
Such trees  are unordered, 
bijectively leaf-labeled 
and 
their internal nodes have at least two children each.
Labels are species under study and the branching pattern of the tree describes the way in which speciation events lead from ancestral species to more recent ones.

\paragraph{Leaf labels.}
For convenience, 
we will identify the leaves with their labels when the tree is understood.
Let $T$ be a (rooted evolutionary) tree.
The leaf label set of $T$ is denoted by $\feuille{T}$.
We say that $T$ is a tree \emph{on} $\feuille{T}$.
The \emph{size} of a tree is defined as the cardinality of its leaf set. 

\paragraph{Degree.}
The \emph{(out-)degree} of a node in $T$ is the number of its children.
The \emph{maximum degree} of $T$, 
denoted by $\Delta(T)$,
is the largest degree over all nodes of $T$.

\paragraph{Parenthetical notation.}
Parenthetical notation is a convenient way to represent evolutionary trees.
Given $d$ non-empty trees $T_1$, $T_2$, \ldots, $T_d$ with pairwise disjoint leaf sets, $\lr{T_1, T_2, \ldots, T_d}$ denotes the tree whose root has degree $d$ and admits as child subtrees $T_1$, $T_2$, \ldots, $T_d$.

\paragraph{Restriction.}
For each subset  $X  \subseteq \feuille{T}$,  the \emph{(topological) restriction} of $T$ to $X$ is denoted by $\restr{T}{X}$.
Colloquially,  
$\restr{T}{X}$ is the tree on $X$ displaying the branching information of $T$ relevant to $X$.
Restriction is formally defined by induction as follows.

\begin{trivlist}
\item[](\emph{Basis}).
For each leaf-tree $\ell$, 
$\restr{\ell}{\{ \ell \}} = \ell$ 
and 
$\restr{\ell}{\emptyset}$ is the empty tree.

\item[](\emph{Inductive step}).
Assume that $T$ is of size at least two:  $T = \lr{T_1, T_2, \ldots, T_d}$ with $d \ge 2$.
If $X$ is a subset of $\feuille{T_i}$ for some $i \in \seg{1}{d}$ then $\restr{T}{X} = \restr{T_i}{X}$,
otherwise, $\restr{T}{X}$ is the tree on $X$ whose root admits as child subtrees all non-empty trees of the form $\restr{T_i}{(\feuille{T_i} \cap X)}$ with $i \in \seg{1}{d}$.
\end{trivlist}

\subsubsection{\pmast{} and \pmct}

Let $\calt$ be a collection of trees on a common leaf set.

\paragraph{Agreement subtree.}
An \emph{\AST} of $\calt$ is a tree $T$ \suchthat, 
$\forall T_i \in \calt$, 
$T = \restr{T_i}{\feuille{T}}$.
The \textsc{Maximum Agreement SubTree} problem (\pmast) consists of finding an \AST{} of $\calt$ of largest size.
In phylogenetics, 
the maximum size of an \AST{} of $\calt$ is a useful measure of the similarity of the trees in $\calt$  \cite{FG85}.
From the point of view of the \pmast{} problem, 
a node $\nu$ of degree $d$ in an input evolutionary tree represents the simultaneous creation of $d$ descendant from the ancestral species represented by $\nu$.
As such events are rare if $d$ is greater than two, 
the trees that people want to calculate maximum \AST{} for have usually small maximum degrees.

\paragraph{Compatible tree.}
Let $T$ and $T'$ be two trees on a common leaf set.
We say that $T$ \emph{refines} $T'$ if $T'$ can be obtained by collapsing a selection of edges of $T$. 
A tree \emph{compatible} with $\calt$ is a tree $T$ \suchthat, 
$\forall T_i \in \calt$, 
$T$ refines $\restr{T_i}{\feuille{T}}$.
Obviously, agreement implies compatibility.
The converse is usually false for collections including at least a non-binary tree.
The \textsc{Maximum Compatible Tree} problem (\pmct) consists of finding a tree of largest size compatible with $\calt$.
The \pmct{} problem is more relevant than the \pmast{} problem when comparing  reconstructed evolutionary trees \cite{HS96,GW01}.
From the point of view of \pmct,
a non-binary node is usually interpreted as a lack of decision with respect to the  relative grouping of its children rather than as a multi-speciation event. 
As data sequences are getting longer and phylogenetic methods more accurate, 
the maximum degree of indecision in reconstructed trees is expected to decrease to a small constant.

\subsubsection{Previous results}
\pmast{} is polynomial on two trees (see \cite{Kao01} for the latest algorithm)
but becomes $\NP$-hard on three input trees \cite{AK97}.
\pmct{} is $\NP$-hard on two trees even if one of them is of maximum degree three \cite{Hein96} (see also \cite{HS96}).

Consider now the general setting of an arbitrary number, denoted by $k$, of input trees.
Let $\calt =  \left\{ T_1, T_2, \ldots T_k \right\}$ 
be the input collection.
Let $n$ be the cardinality of the common leaf set of the $T_i$'s, 
let $d := \min_{i = 1}^k \Delta(T_i)$ 
and 
let $D := \max_{i = 1}^k \Delta(T_i)$.
Above,
we argued about the relevance of solving \pmast{} and \pmct{} on bounded maximum degree  trees.
Three different algorithms were proposed to solve \pmast{} in polynomial time for bounded $d$ \cite{AK97,Fal95,B97}.
The fastest of these algorithms \cite{Fal95,B97} run in $O(n^d + kn^3)$ time.

Besides, 
\pmct{} can be solved in $O(4^{k D}n^k)$ time \cite{GW01}. 
Hence, 
for bounded $k$,
\pmct{} is FPT in $D$.
The same result holds for \pmast.
Assume that a  bound $p$ on the number of leaves to be removed from the input set of leaves so that the input trees agree, resp.~are compatible, is added to  the input. 
Then \pmast, resp. \pmct, can be solved in $O\big(\min\{3^p k n, \alpha^p+kn^3\}\big)$ time, 
where $\alpha$ is a constant less than three  \cite{BerryN06}. 
Thus, both problems are FPT with respect to $p$.

\subsubsection{Our contribution} \label{sec:contrib}
We prove that both \pmast{} and \pmct{} are $\WI$-hard with respect to $D$.
Furthermore, 
let $\varphi : \N \to \N$ be an arbitrary recursive function.
Note that the input $\calt$ is of size $\widetilde O(kn)$.
We prove the following.
\begin{enumerate}[$(R1)$.]
\item 
\pmast{} cannot be solved in $\varphi(D){(kn)}^{o(D)}$ time, 
unless $\SNP \subseteq \SE$.
\item 
\pmct{} cannot be solved in $\varphi(D){(kn)}^{o(2^{D \quotient 2})}$ time,
unless $\SNP \subseteq \SE$. 
\end{enumerate}
Recall that $\SE$ \cite{ImpagliazzoPZ01} is the class of problems solvable in subexponential time and that $\SNP$ \cite{PapadimitriouY91} contains many $\NP$-hard problems. 
Hence,
 the inclusion $\SNP \subseteq \SE$ is unlikely.
According to result~$(R1)$, 
the $O(n^d + kn^3)$ time algorithms for \pmast{} \cite{Fal95,B97}  are somehow optimum.
Results $(R1)$ and $(R2)$ are proved in Sections~\ref{sec:MASThard} and~\ref{sec:MCThard}, respectively.

\subsection{Parameterized complexity} \label{sec:par-cpx}

In order to clearly prove our intractability results,
we recall the main concepts of parameterized complexity \cite{DowFel99}, 
together with some recent results.
We also introduce the notions of linear FPT-reduction and weak fixed-parameter tractability.    

Let $\Sigma$ be a finite alphabet. 
The set of all finite words over $\Sigma$ is denoted by $\Sigma^\star$, 
and for each word $x \in \Sigma^\star$,
$| x |$ denotes the \emph{length} of $x$.
A \emph{parameterized (decision) problem} is a subset $P \subseteq \N \times \Sigma^\star$.
Each element of $(k, x) \in \N \times \Sigma^\star$ is an \emph{instance} of $P$, $k$ standing for the \emph{parameter}.
A \emph{yes-instance} of $P$ is an element of $P$ 
and 
a \emph{no-instance} of $P$ is an element of $(\N \times \Sigma^\star) \setnabot P$.

\subsubsection{Fixed-parameter tractability and weak fixed-parameter tractability}

The parameterized problem $P$ is called \emph{fixed-parameter tractable} (FPT), 
 if 
there exist an algorithm $A$ and a recursive function $\varphi : \N \to \N$ such that, 
for each $(k, x) \in \N \times \Sigma^\star$,
 $A$ decides whether  $(k, x)$ is a yes-instance of $P$ in $\varphi(k) {| x |}^{O(1)}$ time.
The parameterized problem $P$ is called \emph{weakly fixed-parameter tractable} (WFPT) if 
there exist an algorithm $A$ and a recursive function $\varphi : \N \to \N$ such that, 
for each $(k, x) \in \N \times \Sigma^\star$,
 $A$ decides whether  $(k, x)$ is a yes-instance of $P$ in 
$\varphi(k) {| x |}^{o(k)}$ time.

\subsubsection{FPT-reduction and linear FPT-reduction}

Let $P$, $Q \subseteq  \N \times \Sigma^\star$ be two parameterized problems and let $f :\N \times \Sigma^\star \to \N \times \Sigma^\star$.

We say that $f$ is a (many-to-one, strongly uniform) \emph{FPT-reduction} from $P$ to $Q$ if there exist 
recursive functions $g : \N \times \Sigma^\star \to \Sigma^\star$ 
and
$\varphi$, $\gamma :  \N \to \N$ satisfying,
for all $(k, x) \in \N \times \Sigma^\star$:
\begin{enumerate}
\item 
$f(k, x)$ is computable in $\varphi(k) {| x |}^{O(1)}$ time,
\item 
$f(k, x) \in Q$ if and only if $(k, x) \in P$, and 
\item 
$f(k, x)  = (\gamma(k), g(k, x))$.
\end{enumerate}
Moreover, 
if $\gamma$ is at most linearly increasing (\ie{} if $\gamma(k) = O(k)$ as $k \to \infty$) then we say that $f$ is a \emph{linear FPT-reduction} from $P$ to $Q$.

\smallskip 

FPT-reductions compose, and preserve fixed-parameter tractability.
Linear FPT-reductions compose, and preserve weak fixed-parameter tractability.
Note that our notion of linear FPT-reduction is slightly different from the one introduced by 
Chen, Huang, Kanj and Xia \cite{ChenHKX06}.

\subsubsection{Independent set}

Formally,
an (undirected) \emph{graph} is an ordered pair $G =(V, E)$ where $V$ is a finite set of \emph{vertices} and where $E$ a set of $2$-element subsets of $V$.
The elements of $E$ are the \emph{edges} of $G$. 
The elements of an edge are called its \emph{endpoints}.
An \emph{independent set} of $G$ is a subset $I \subseteq V$ \suchthat, 
for each edge $e \in E$, 
at least one of its endpoint is not in $I$.
The problem of finding an independent set of maximum cardinality in a given input graph plays a central role in computational complexity theory.
\defdecpb{\textsc{Independent Set} (\pis).}
{A positive integer $k$ and a graph $G = (V, E)$.}
{Is there an independent set of $G$ with cardinality $k$?} 
The version of \pis{} parameterized by $k$ is denoted  by \pisk.
This problem is not believed to be FPT as it is complete under FPT-reductions for the class $\WI$ \cite{DowFel99}.
Moreover, \pisk{} is not WFPT either, 
unless $\SNP \subseteq \SE$ \cite[Theorem~5.5]{ChenHKX06}.

\section{Parameterized complexity of \pmast}
\label{sec:MASThard}

The decision version of the  \pmast{} problem is:
\defdecpb{\textsc{Agreement SubTree} (\past).}
{An integer $q \ge 1$ and a finite collection $\calt$ of trees on a common leaf set.}
{Is  there an \AST{} of $\calt$ with  size $q$?}
We denote by \pastD{} the version of \past{} parameterized by $D := \max_{T \in \calt} 
\Delta(T)$.  
In this section, we prove: that \pastD{} is $\WI$-hard, and Result~$(R1)$ stated in Section~\ref{sec:contrib}.
According to Section~\ref{sec:par-cpx}, 
it is sufficient to present a linear FPT-reduction from  \pisk{} to \pastD.

For each integer $p \ge 1$, 
we introduce the following auxiliary problem:
\defdecpb{\textsc{Partitioned Independent Set with multiplicity} $p$ (\ppisp).}
{An integer $k \ge 1$, a graph $G = (V, E)$, and $k$ independent sets  $V_1$, $V_2$, \ldots, $V_k$  of $G$ of equal cardinality partitioning $V$.}
{Is there an independent set $I$ of $G$ \suchthat{} $I \cap V_i$ has cardinality $p$ for all $i \in \seg{1}{k}$?}
For each instance $(k, G, V_1, V_2, \ldots, V_k)$ of \ppisp{}, 
the graph $G$ is $k$-colorable: 
the $V_i$'s yield a $k$-coloring of $G$.
The version of \ppisp{} parameterized by $k$ is denoted by \ppispk.
We reduce \pisk{}  to \pastD{} going through \ppisk.
In the next section, 
the decision version of \pmct{} is reduced to \pis{} going through \ppisdeux.

\begin{lemma} \label{lem:PIS}
\pisk{} \lfptred{}s to  \ppisk.
\end{lemma}

\begin{proof}
Reduce \pisk{} to \ppisk{} in the same way as Pietrzak  reduces \textsc{Clique} to \textsc{Partitioned Clique}  \cite{Pietrzak-03}.
Each instance $(k, G)$ of \pis{} is transformed into an instance $(k, \tilde G, \tilde V_1, \tilde V_2, \ldots, \tilde V_k)$ of \ppis{} where $\tilde G$ and the $\tilde V_i$'s are as follows.

Let $V$ denote the vertex set of $G$.
$\tilde G$ is the graph on $V \times \seg{1}{k}$ whose edge set is given by: 
for all $(u, i)$, $(v, j) \in V \times \seg{1}{k}$, 
$\left\{ (u, i), (v, j) \right\}$ is an edge of $\tilde G$ 
if and only if 
$i \ne j$ and
 either  $\{ u, v \}$ is an edge of $G$ or $u = v$.
For each $i \in \seg{1}{k}$, 
$\tilde V_i$ is defined as $\tilde V_i := V \times \{ i \}$.
It is clear that $(k, \tilde G, \tilde V_1, \tilde V_2, \ldots, \tilde V_k)$ is an instance of \ppisk{} computable  from $(k, G)$ in polynomial time.
It remains to check that 
$(k, G)$ is a yes-instance of \pis{}
if and only if  
$(k, \tilde G, \tilde V_1, \tilde V_2, \ldots, \tilde V_k)$  is a yes-instance of \ppis.

\begin{trivlist}
\itemif{}
Assume there exists an independent set $\tilde I$ of $\tilde G$ \suchthat{} $\tilde I \cap \tilde V_i$ is a singleton for all $i \in \seg{1}{k}$.
For each $i \in \seg{1}{k}$, 
let $v_i \in V_i$ be \suchthat{} $\tilde I \cap \tilde V_i= \{ (v_i, i) \}$.
The set $I := \{ v_1, v_2, \ldots, v_k \}$ is an independent set of $G$ with cardinality $k$.

\itemonlyif{}
Conversely, assume that there exists an independent set $I$ of $G$ with cardinality $k$.
Write $I$ in the form $I = \left\{ v_1, v_2, \ldots, v_k \right\}$.
The set $\tilde I := \left\{ (v_1, 1), (v_2, 2), \ldots, (v_k, k) \right\}$ is an independent set of $\tilde G$ and $\tilde I \cap \tilde V_i = \{ (v_i, i) \}$ is a singleton for all $i \in \seg{1}{k}$.
\end{trivlist}
\end{proof}

In order to clearly prove Theorem~\ref{thm:pis-MAST}, 
we first introduce some useful vocabulary.
 
\begin{definition}
Let $T$ and $T'$ be two trees and let $L$ be a subset of  $\feuille{T} \cap \feuille{T'}$.
We say that $T$ and $T'$ \emph{disagree} on $L$ if $\restr{T}{L}$ and $\restr{T'}{L}$ are distinct.
\end{definition}
Assume that $\feuille{T} \subseteq \feuille{T'}$.
If there exists a subset $L \subseteq \feuille{T}$ \suchthat{} $T$ and $T'$ disagree on $L$ then $T$ is not a restriction of $T'$.
Conversely, 
if $T$ is  not a  restriction of $T'$ 
then $T$ and $T'$ disagree on some $3$-element subset of $\feuille{T}$ \cite{B97}. 
This explains the central role played by $3$-leaf sets of disagreement in the proofs of Lemmas~\ref{lem:MASTcontrol} and~\ref{lem:MASTselection} below.
Note that given three distinct leaf labels $a$, $b$ and $c$, 
there are exactly four distinct trees on $\{ a, b, c \}$: 
the non-binary tree $\lr{a, b, c}$,
and the three binary trees $\lr{\lr{b, c}, a}$, $\lr{\lr{a, c}, b}$ and $\lr{\lr{a, c}, b}$.

\begin{theorem} \label{thm:pis-MAST}
\pisk{} \lfptred{}s to \pastD.
\end{theorem}

\begin{proof}
According to Lemma~\ref{lem:PIS}, it suffices to \lfptred{}  \ppisk{} to \pastD.
Each instance  $(k, G, V_1, V_2, \ldots, V_k)$ of \ppis{} is transformed into an instance 
$(q, \calt)$ of \past{}  where $q := k$ and where $\calt$ is a collection of trees described below.
Without loss of generality, 
we can assume that all $V_i$'s ($i \in \seg{1}{k}$) have cardinality at least three and that $k$ is at least three.

\paragraph{The collection $\calt$.}
We construct a  collection $\calt$ of gadget trees whose leaf set is the vertex set $V :=  V_1 \cup V_2 \cup \cdots \cup V_k$ of $G$.

For each $i \in \seg{1}{k}$, 
compute an arbitrary \emph{binary} tree $B_i$ on $V_i$.
The tree on $V$ whose root admits $B_1$, $B_2$, \ldots, $B_k$ as child subtrees is denoted by $C$: $C = \lr{B_1, B_2, \ldots, B_k}$.
Every tree of $\calt$ is obtained by modifying the positions of exactly two leaves of $C$.

For all $a$, $b \in V$ with $a \ne b$, 
$C_{a, b}$ denotes the tree on $V$ obtained from $C$, 
by first removing its leaves $a$ and $b$,
and then re-grafting both of them as new children of the root.
Formally, 
$C_{a, b}$ is the tree
$$  
\lr{
\restr{B_1}{(V_1 \setnabot \{ a, b \})}, 
\restr{B_2}{(V_2 \setnabot \{ a, b \})}, 
\ldots, 
\restr{B_k}{(V_k \setnabot \{ a, b \}), a, b}
} \, .
$$
We set $\calc := \{ C \} \cup \{ C_{a, b} : a, b \in V, a \ne b \}$.

\begin{remark}
There exist at most two indices $i$ \suchthat{} $\restr{B_i}{(V_i \setnabot \{ a, b \})}$ is distinct from $B_i$, 
and since $V_i$ has cardinality at least three, 
 $\restr{B_i}{(V_i \setnabot \{ a, b \})}$ is a non-empty tree for all $i$.
\end{remark}

Let $E$ denote the edge set of $G$: 
$G = (V, E)$.
For each edge $e = \{ a, b \} \in E$,  
$S_e$ denotes the tree on $V$ obtained from $C$,
by first removing its leaves $a$ and $b$, 
and then re-grafting $\lr{a, b}$ as a new child of the root.
Formally, $S_e$ is the tree
$$
\lr{
\restr{B_1}{(V_1 \setnabot e)}, 
\restr{B_2}{(V_2 \setnabot e)}, 
\ldots, 
\restr{B_k}{(V_k \setnabot e)}, 
\lr{a, b}
} \, .
$$
The collection of trees $\calt$ is defined as $\calt := \calc \cup  \left\{ S_e : e \in E \right\}$ (see Figure~\ref{fig:gadgetMAST}):
$\calc$ is the \emph{control component} of our gadget
and 
the $S_e$'s ($e \in E$) are its \emph{selection components}.
\begin{figure}
\psfrag{a}[][]{$\mathtt{a}$}
\psfrag{b}[][]{$\mathtt{b}$}
\psfrag{c}[][]{$\mathtt{c}$}
\psfrag{d}[][]{$\mathtt{d}$}
\psfrag{e}[][]{$\mathtt{e}$}
\psfrag{f}[][]{$\mathtt{f}$}
\psfrag{g}[][]{$\mathtt{g}$}
\psfrag{h}[][]{$\mathtt{h}$}
\psfrag{i}[][]{$\mathtt{i}$}
\psfrag{j}[][]{$\mathtt{j}$}
\psfrag{k}[][]{$\mathtt{k}$}
\psfrag{l}[][]{$\mathtt{l}$}
\psfrag{BUN}[][B]{$B_1$}
\psfrag{BDEUX}[][B]{$B_2$}
\psfrag{BTROIS}[][B]{$B_3$}
\psfrag{C}[][B]{$C$}
\psfrag{Cki}[][B]{$C_{\mathtt{i}, \mathrm{k}}$}
\psfrag{Scf}[][B]{$S_{\{ \mathtt{c}, \mathtt{f} \}}$}
\centering
\epsfig{figure=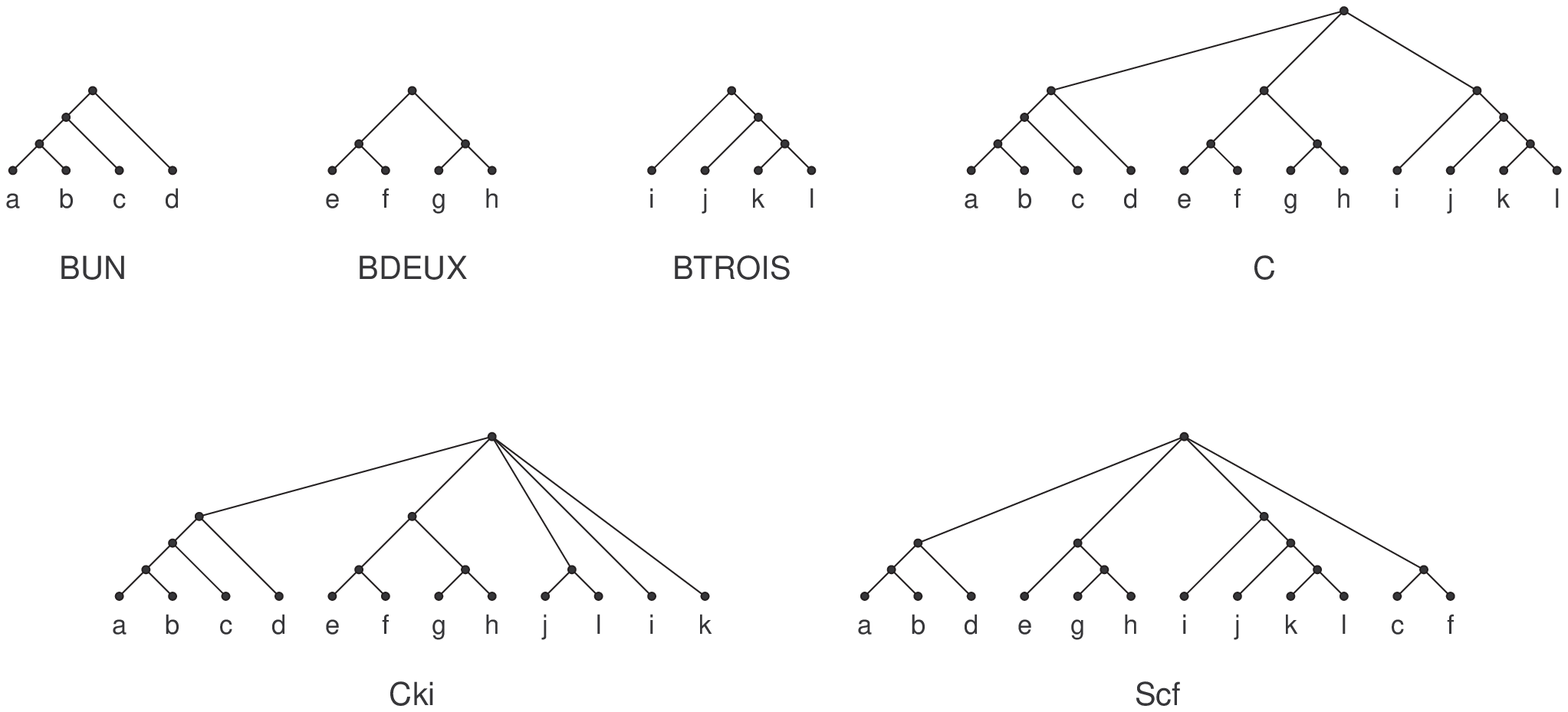,width=\linewidth}
\caption{Some of the gadget trees encoding an instance $(k, G, V_1, V_2, \ldots, V_k)$ of \ppisk{}
where
$k = 3$, 
$V_1 = \{ \mathtt{a}, \mathtt{b}, \mathtt{c}, \mathtt{d} \}$,
$V_2 = \{ \mathtt{e}, \mathtt{f}, \mathtt{g}, \mathtt{h} \}$,
$V_3 = \{ \mathtt{i}, \mathtt{j}, \mathtt{k}, \mathtt{l} \}$
and 
$\{ \mathtt{c}, \mathtt{f} \}$ is an edge of $G$.
\label{fig:gadgetMAST}
}
\end{figure}

\begin{lemma}[Control] \label{lem:MASTcontrol}
Let $T$ be a tree with  $\feuille{T} \subseteq V$.
Statements $(i)$ and $(ii)$ below are equivalent.
\begin{enumerate}[$(i)$.]
\item $T$ is an \AST{} of  $\calc$ with size $k$.
\item $T = \lr{c_1, c_2, \ldots, c_k}$ for some $(c_1, c_2, \ldots, c_k) \in V_1 \times V_2 \times \cdots \times  V_k$.
\end{enumerate}
\end{lemma}

\begin{proofqed}
Let $(c_1, c_2, \ldots, c_k) \in V_1 \times V_2 \times \cdots \times  V_k$.
Distinct $c_i$'s appear in distinct child subtrees of the root of $C$, resp.~of $C_{a, b}$.
Hence, $\lr{c_1, c_2, \ldots, c_k}$ is a restriction of $C$, resp.~of $C_{a, b}$.
This proves that $(ii)$ implies $(i)$. 
It remains to show that $(i)$ implies $(ii)$.

Assume $(i)$: $T$ is an \AST{} of $\calc$ with size $k$.

\begin{trivlist}
\trivpoint
We first prove that $T$ has height one.
By  way of contradiction,
suppose that the height of $T$ is greater than one.
Then,
one can find three distinct leaves $a$, $b$, $c \in \feuille{T}$ \suchthat{} 
$\restr{T}{\{ a, b, c \}} = \lr{\lr{a, b}, c}$.
(Indeed, there exists an internal non-root node $\nu$ of $T$.
Pick a leaf $c$ which is not a descendant of $\nu$
and  two descendant leaves $a$ and $b$ of $\nu$.)
However, $\restr{C_{a, b}}{\{ a, b, c \}} = \lr{a, b, c}$, 
and thus $T$ and $C_{a, b}$ disagree on $\{ a, b, c \}$: 
contradiction.
\end{trivlist}

Since $T$ has height one,
there exist $k$ pairwise distinct leaf labels $c_1$, $c_2$, \ldots, $c_k \in V$ 
\suchthat{} 
$T = \lr{c_1, c_2, \ldots, c_k}$.

\begin{trivlist}
\trivpoint 
We now show that distinct $c_j$'s belong to distinct $V_i$'s.
By way of contradiction, 
assume there exist three indices $i$, $j_1$, $j_2 \in \seg{1}{k}$ satisfying  $j_1 \ne j_2$,
$c_{j_1} \in V_i$ and $c_{j_2} \in V_i$.
Since $k$ is greater that two, 
there exists  $j \in \seg{1}{k}$ such that $j \notin \{ j_1, j_2 \}$.
If $c_j \in V_i$  then 
$\restr{C}{\{ c_{j_1}, c_{j_2}, c_j \}} = \restr{B_i}{\{ c_{j_1}, c_{j_2}, c_{j} \}}$
and
if $c_j \notin V_i$ then
$\restr{C}{\{ c_{j_1}, c_{j_2}, c_{j} \}} = \lr{\lr{c_{j_1}, c_{j_2}}, c_{j}}$.
In both cases, $\restr{C}{\{ c_{j_1}, c_{j_2}, c_j \}}$ is a binary tree unlike $\restr{T}{\{ c_{j_1}, c_{j_2}, c_j \}}$.
Thus, $C$ and $T$ disagree on $\{ c_{j_1}, c_{j_2}, c_j \}$: 
contradiction.
\end{trivlist}

Up to a permutation of the $c_i$'s, one has $(c_1, c_2, \ldots, c_k) \in V_1 \times V_2 \times \cdots \times  V_k$. 
This proves $(ii)$ and  concludes the proof of Lemma~\ref{lem:MASTcontrol}.
\end{proofqed}

\begin{lemma}[Selection] \label{lem:MASTselection}
Let $e \in E$ be an edge of $G$ 
and
let $(c_1, c_2, \ldots, c_k) \in V_1 \times V_2 \times \cdots \times V_k$.
The tree 
$\lr{c_1, c_2, \ldots, c_k}$ is a restriction of $S_e$
if and only if 
at least one endpoint of $e$ is not in $\{ c_1, c_2, \ldots, c_k \}$.
\end{lemma}

\begin{proofqed}
The ``if part'' is easy.
Let us now show the ``only if'' part.

Assume that $\lr{c_1, c_2, \ldots, c_k}$ is a restriction of $S_e$ and that 
$e \subseteq \{ c_1, c_2, \ldots, c_k \}$.
Let $c_{i_1}$ and $c_{i_2}$ be the two endpoints of $e$:
$e = \{ c_{i_1}, c_{i_2} \}$.
Since $k$ is greater than two,
 there exists $i \in \seg{1}{k}$ \suchthat{} $c_i \notin e$.
 The restriction of $S_e$ to  $\{ c_{i_1}, c_{i_2}, c_i \}$ equals $\lr{\lr{c_{i_1}, c_{i_2}}, c_i}$, 
and thus $S_e$ disagrees with $\lr{c_1, c_2, \ldots, c_k}$ on $\{c_{i_1}, c_{i_2}, c_i \}$: contradiction.
This concludes the proof of Lemma~\ref{lem:MASTselection}.
\end{proofqed}

\paragraph{Correctness of the reduction.}
It is clear that $(q, \calt)$ is computable in polynomial time from $(k, G, V_1, V_2, \ldots, V_k)$.
Moreover, the root of $C$ has degree $k$,
the root of $C_{a, b}$ has degree $k + 2$, 
the root of $S_e$ has degree $k + 1$, 
and any non-root internal node of a tree in $\calt$ has degree two.
Hence, 
the maximum degree $D$ over  all trees in $\calt$ is equal to $k + 2$:
$D = O(k)$.
Eventually, 
let us derive from  Lemmas~\ref{lem:MASTcontrol} and~\ref{lem:MASTselection} that 
$(k, G, V_1, V_2, \ldots, V_k)$ is a yes-instance of \ppis{} 
if and only if 
$(q, \calt)$  is a yes-instance of \past.
\begin{trivlist}
\itemif
Assume there exists an \AST{} $T$ of $\calt$ with size $q = k$.
The tree $T$ is of the form $T = \lr{c_1, c_2, \ldots, c_k }$
for some 
$(c_1, c_2, \ldots, c_k) \in V_1 \times V_2 \times \cdots \times V_k$ by Lemma~\ref{lem:MASTcontrol}.
Furthermore, 
the set $I := \{ c_1, c_2, \ldots, c_k \}$ is an independent set of $G$ by Lemma~\ref{lem:MASTselection}, 
and for every $i \in \seg{1}{k}$, 
$I \cap V_i = \{ c_i \}$ is a singleton. 
\itemonlyif 
Conversely,
assume that there exists an independent set $I$ of $G$ \suchthat{} $I \cap V_i$ is a singleton for all $i \in \seg{1}{k}$.
Write $I$ in the form $I = \left\{ c_1, c_2, \ldots, c_k \right\}$ with 
$(c_1, c_2, \ldots, c_k) \in V_1 \times V_2 \times \cdots \times V_k$.
The tree $\lr{c_1, c_2, \ldots, c_k }$ is both an agreement subtree of $\calc$ by Lemma~\ref{lem:MASTcontrol} and  an agreement subtree of $\left\{ S_e : e \in E \right\}$ by Lemma~\ref{lem:MASTselection}.
Therefore, 
$\lr{c_1, c_2, \ldots, c_k }$  is  an \AST{} of $\calt$ with size $q$.
\end{trivlist}
\end{proof}

\section{Parameterized complexity of \pmct} \label{sec:MCThard}

The decision version of  the \pmct{} problem is:
\defdecpb{\textsc{Compatible Tree} (\pct).}
{An integer $q \ge 1$ and a finite collection $\calt$ of trees on a common leaf set.}
{Is  there a tree of size $q$ compatible with  $\calt$?}
Let  \pctD{} denote the version of \pct{} parameterized by $2^{\left\lfloor D \quotient 2 \right\rfloor}$, 
where $D := \max_{T \in \calt} \Delta(T)$.  
In this section, we  \lfptred{} \pisk{} to \pctD{} in order to prove: 
the $\WI$-hardness of the version of \pct{} parameterized by $D$, 
and Result~$(R2)$ stated in Section~\ref{sec:contrib}.
\ppisdeux{} is used as an auxiliary problem.

\begin{lemma} \label{lem:PPISdeux}
\pisk{} \lfptred{}s to \ppisdeuxk. 
\end{lemma}

\begin{proof}
According to Lemma~\ref{lem:PIS}, 
it suffices to \lfptred{} \ppisk{} to \ppisdeuxk.
We rely on a padding argument.
Each instance $(k, G, V_1, V_2, \ldots, V_k)$ of \ppis{} is transformed into an instance 
$(k, \tilde G, \tilde V_1, \tilde V_2, \ldots, \tilde V_k)$ of \ppisdeux{} where $\tilde G$ and the $\tilde V_i$'s are as follows.

Informally, 
$\tilde G$ is obtained by adding $k$ isolated vertices to $G$,
and each $\tilde V_i$ is obtained by adding a single one of these new vertices to $V_i$.
More formally, 
let $V$ denote the vertex set of $G$ 
and 
let $E$ denote the edge set of $G$:
$V =  V_1 \cup V_2 \cup \cdots \cup  V_k$ and  $G = (V, E)$.
Let $a_1$, $a_2$, \ldots, $a_k$ be $k$ new vertices: 
for all $i$, $j \in \seg{1}{k}$,
 $a_i$ is not an element of $V$,
 and $i \ne j$ implies $a_i \ne a_j$. 
Construct 
$\tilde G := ( V  \cup \{ a_1, a_2, \ldots, a_k \}, E )$,
and $\tilde V_i := V_i \cup \{ a_i \}$ for each $i \in \seg{1}{k}$.

It is clear that 
$(k, \tilde G, \tilde V_1, \tilde V_2, \ldots, \tilde V_k)$ is an instance of \ppisdeux{} computable in polynomial time from $(k, G, V_1, V_2, \ldots, V_k)$.
It remains to check that
$(k,  G,  V_1,  V_2, \ldots,  V_k)$ is a yes-instance of \ppis{} 
 if and only if 
$(k, \tilde G, \tilde V_1, \tilde V_2, \ldots, \tilde V_k)$  is a yes-instance of \ppisdeux.

\begin{trivlist}
\itemonlyif{}
Assume that there exists an independent set $I$ of $G$ \suchthat{} $I \cap V_i$ is a singleton for every $i \in \seg{1}{k}$.
Then $\tilde I := I \cup \{ a_1, a_2, \ldots, a_k \}$ is an independent set of $\tilde G$,
and $\tilde I \cap \tilde V_i$ is a doubleton for all $i \in \seg{1}{k}$.
\itemif{}
Conversely, assume that there exists an independent set $\tilde I$ of $\tilde G$ \suchthat{}
$\tilde I \cap \tilde V_i$ is a doubleton for every  $i \in \seg{1}{k}$.
For each $i \in \seg{1}{k}$, 
pick an element $v_i$ in $\tilde I \cap \tilde V_i$ distinct from $a_i$.
The set $I := \{ v_1, v_2, \ldots, v_k \}$  is an independent set of $G$,
and $I \cap V_i = \{ v_i \}$ is a singleton for all $i \in \seg{1}{k}$. 
\end{trivlist}
\end{proof}

\begin{remark}
The mapping 
$(k,  G,  V_1,  V_2, \ldots,  V_k) 
\longmapsto 
(k, \tilde G, \tilde V_1, \tilde V_2, \ldots, \tilde V_k)$, 
presented in  the proof of Lemma~\ref{lem:PPISdeux},
 induces a linear FPT-reduction from \ppispk{} to  \ppisqk{} for \emph{any} integer $p \ge 1$,.
Since \pisk{} \lfptred{}s to \ppisk{} by Lemma~\ref{lem:PIS}, 
\pisk{} \lfptred{}s to \ppispk{} for every  integer $p \ge 1$.
\end{remark}

Definitions~\ref{def:composetree}, \ref{def:Rn} and~\ref{def:Hn} introduce gadgets that are  used to reduce \ppisdeux{} to \pct{} in the the proof of Theorem~\ref{thm:pis-MCT}.

\begin{definition} \label{def:composetree}
Let $n$ be a positive integer, 
let $T$ be a tree on $\seg{1}{n}$, 
and let $T_1$, $T_2$, \ldots, $T_n$ be $n$ non-empty trees with pairwise disjoint leaf sets. 
The tree on $\feuille{T_1} \cup \feuille{T_2} \cup \cdots \cup \feuille{T_n}$,
obtained by replacing each leaf $i$ of $T$ with  $T_i$ is denoted by $T[T_1, T_2, \ldots, T_n]$.
\end{definition}
For instance, let  $T := \lr{\lr{1, 2}, \lr{3, \lr{4, 5}}, 6}$.
For any non-empty trees  $T_1$, $T_2$, $T_3$, $T_4$, $T_5$, $T_6$ with pairwise disjoint leaf sets,  
we have 
$$T[T_1, T_2, T_3, T_4, T_5, T_6] = \lr{\lr{T_1, T_2}, \lr{T_3, \lr{T_4, T_5}}, T_6} \, ,
$$
 and in particular, 
$$
T[\lr{2, 3}, 1, \lr{6, 7, 8}, 4, 5, \lr{\lr{9, 11}, 10}]
=  
\lr{\lr{1, \lr{2, 3}}, \lr{\lr{4, 5}, \lr{6, 7, 8}}, \lr{\lr{9, 11}, 10}} \, .
$$

\begin{definition}\label{def:Rn}
For each integer $n \ge 1$,
$R_n$ denotes the binary tree on $\seg{1}{n}$, 
defined recursively as follows: 
\begin{itemize}
\item
$R_1 = 1$, and 
\item
 $R_n = \lr{R_{n - 1}, n}$  for every integer $n \ge 2$.
\end{itemize}
\end{definition}
For instance,  one has
$R_2 = \lr{1, 2}$,
$R_3 = \lr{\lr{1, 2}, 3}$, 
$R_4 = \lr{\lr{\lr{1, 2}, 3}, 4}$, 
$R_5 =  \lr{\lr{\lr{\lr{1, 2}, 3}, 4}, 5}$,
\emph{etc}.

\begin{property} \label{prop:rateau}
Let $n$ be a positive integer.
Let $v^1$, $v^2$, \ldots, $v^n$ be $n$ pairwise distinct labels.
A tree with leaf labels in $\{ v^1, v^2, \ldots, v^n \}$ is compatible with $\left\{ R_n[v^1, v^2, \ldots, v^n], R_n[v^n, \ldots,  v^2,  v^1] \right\}$ if and only if its size is at most two.
\end{property}
\begin{definition}\label{def:Hn}
For every  integer $k \ge 1$, 
$H_k$ denotes a binary tree on $\seg{1}{k}$ with minimum height $\ceil{\log k}$;
for all  $i$, $j \in \seg{1}{k}$, $H_k^{i, j}$ denotes the tree on $\seg{1}{k}$ obtained from $H_k$ by collapsing all internal edges on the path connecting $i$ and $j$;
$\lambda_k^{i, j}$ denotes the least common ancestor of $i$ and $j$ in $H_k^{i, j}$.
\end{definition}
For instance, 
$\lr{\lr{\lr{1, 2}, \lr{3, 4}}, 5}$ is a suitable tree $H_5$,
and 
$\lr{\lr{\lr{1, 2}, \lr{3, 4}}, \lr{ \lr{5,6}, \lr{7,8}}}$ is a suitable tree $H_8$;
for such trees, 
one has 
$H_5^{1, 4} = \lr{ \lr{1, 2, 3, 4}, 5}$ 
and 
$H_8^{3, 5} = \lr{ \lr{1, 2}, \lr{3, 4, 5, 6}, \lr{7, 8}}$.

\begin{property} \label{prop:deg-lambda}
All internal nodes in $H_k^{i, j}$  are of degree two,
except maybe $\lambda_k^{i, j}$ whose degree is at most $2 \ceil{\log k}$. 
\end{property}

\begin{theorem} \label{thm:pis-MCT}
\pisk{} \lfptred{}s to \pctD.
\end{theorem}

\begin{proof}
According to Lemma~\ref{lem:PPISdeux}, it suffices to \lfptred{} \ppisdeuxk{} to \pctD.
Each instance $(k, G, V_1, V_2, \ldots, V_k)$ of \ppisdeuxk{} is transformed into an instance $(q, \calt)$ of \pct{} where $q := 2k$ and where $\calt$ is a collection of trees described below.

\paragraph{The collection $\calt$.}
We construct a collection $\calt$ of gadget trees on the vertex set $V :=  V_1 \cup V_2 \cup \cdots \cup V_k$ of $G$.
Let $n$ be \suchthat{} $V_i$ has cardinality $n$ for every $i \in \seg{1}{k}$.
For each $i \in \seg{1}{k}$, 
write $V_i$ in the form  $V_i = \left\{ v^1_i, v^2_i, \ldots, v^n_i \right\}$;
$B_i := R_n[v^1_i, v^2_i, \ldots, v^n_i]$ 
and 
$\tilde B_i := R_n[v^n_i, \ldots, v^2_i,  v^1_i]$ encode $V_i$.

Let $C := H_k[B_1, B_2, \ldots, B_k]$ and let $\tilde C := H_k[\tilde B_1, \tilde B_2, \ldots, \tilde B_k]$ (see Figure~\ref{fig:CtildeC}): 
$C$ and $\tilde C$ are the \emph{control components} of our gadget.
\psfrag{Bi}[][r]{$B_i$}
\psfrag{v11}[][]{$v_1^1$}
\psfrag{v12}[][]{$v_1^2$}
\psfrag{v13}[][]{$v_1^3$}
\psfrag{v14}[][]{$v_1^4$}
\psfrag{v21}[][]{$v_2^1$}
\psfrag{v22}[][]{$v_2^2$}
\psfrag{v23}[][]{$v_2^3$}
\psfrag{v24}[][]{$v_2^4$}
\psfrag{v31}[][]{$v_3^1$}
\psfrag{v32}[][]{$v_3^2$}
\psfrag{v33}[][]{$v_3^3$}
\psfrag{v34}[][]{$v_3^4$}
\psfrag{v41}[][]{$v_4^1$}
\psfrag{v42}[][]{$v_4^2$}
\psfrag{v43}[][]{$v_4^3$}
\psfrag{v44}[][]{$v_4^4$}
\psfrag{v51}[][]{$v_5^1$}
\psfrag{v52}[][]{$v_5^2$}
\psfrag{v53}[][]{$v_5^3$}
\psfrag{v54}[][]{$v_5^4$}
\begin{figure}
\psfrag{C}[][]{$C$}
\psfrag{tildeC}[][]{$\tilde C$}
\psfrag{HCINQ}[][r]{$H_5$}
\psfrag{tildeBi}[B][r]{$\tilde B_i$}
\centering 
\epsfig{figure=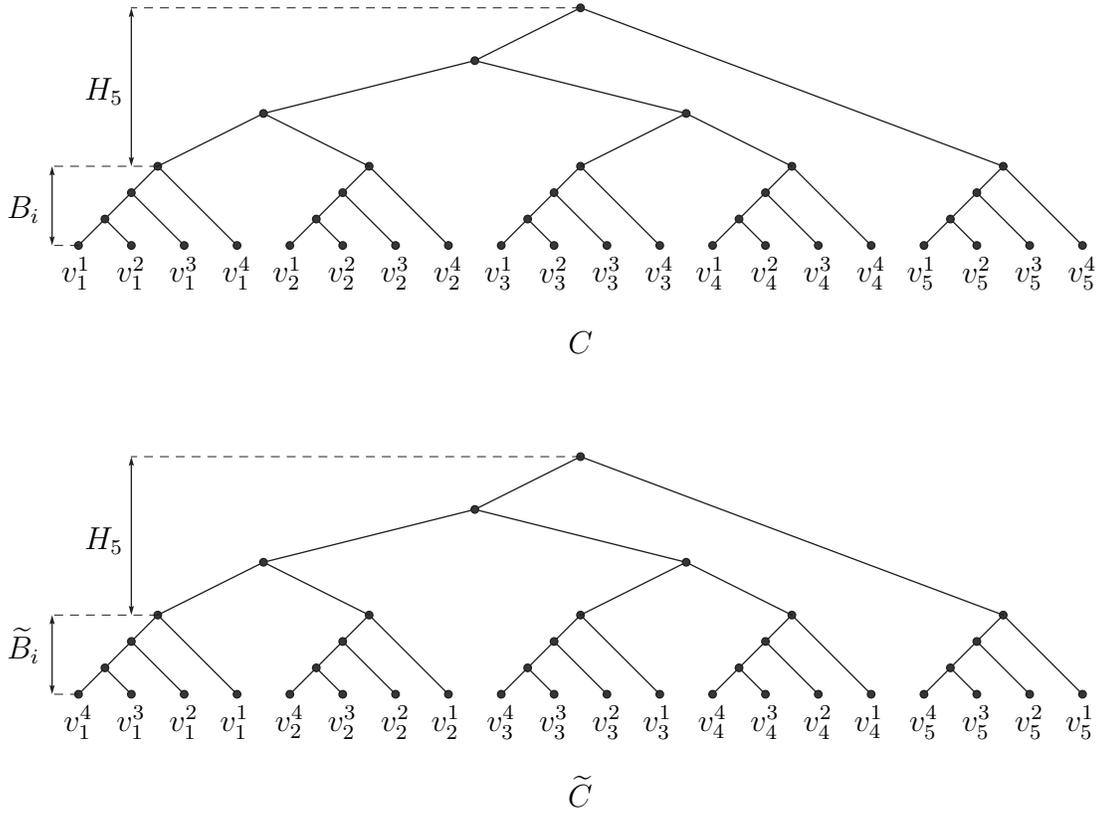,width=\linewidth}
\caption{The trees $C$ and $\tilde C$ in the case of $k = 5$ and $n = 4$. \label{fig:CtildeC}}
\end{figure}

Let $E$ be the edge set of $G$: $G = (V, E)$.
For each edge $e = \{ v_i^r, v_j^s \} \in E$, 
compute the tree $S_e$ obtained from $H_k^{i, j}\left[B_1, B_2,  \ldots, B_k \right]$ by first removing its leaves $v_i^r$ and $v_j^s$, 
and then re-grafting $\lr{v_i^r, v_j^s}$ as a new child subtree of $\lambda_k^{i, j}$ (see Figure~\ref{fig:Se}). 
The $S_e$'s ($e \in E$) are the \emph{selection components} of our gadget.
\begin{figure}
\psfrag{S}[][]{$S_{\left\{ v^2_1, v^3_4 \right\}}$}
\psfrag{H514}[][r]{$H_5^{1, 4}$}
\psfrag{L14}[][r]{$\lambda_5^{1, 4}$}
\centering
\epsfig{figure=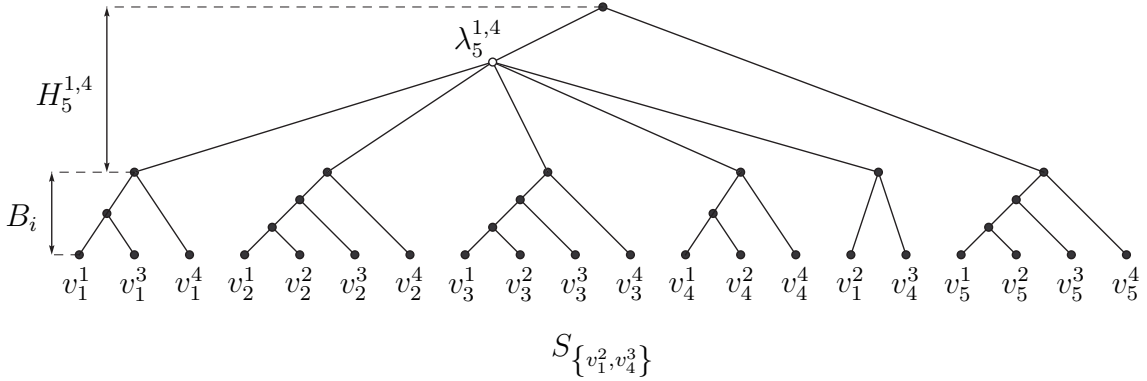,width=\linewidth}
\caption{The tree $S_{\left\{ v^2_1, v^3_4 \right\}}$ in the case of 
$k = 5$ 
and 
$n = 4$.  \label{fig:Se}}
\end{figure}

The collection of trees $\calt$ is defined as $\calt := \{ C, \tilde C \} \cup  \left\{ S_e : e \in E \right\}$.

Property~\ref{lem:MCTcontrol} below is easily deduced from Property~\ref{prop:rateau}.
\begin{property}[Control] \label{lem:MCTcontrol}
Let $T$ be a tree with  $\feuille{T} \subseteq V$.
Statements $(i)$ and $(ii)$ below are equivalent.
\begin{enumerate}[$(i)$.]
\item $T$ is a tree of size $q$,
 compatible with $\{ C, \tilde C \}$.
\item $T$ is of the form $T = H_k[\lr{a_1, b_1}, \lr{a_2, b_2}, \ldots, \lr{a_k, b_k}]$ where, 
for each $i \in \seg{1}{k}$, 
$a_i$ and $b_i$ are two distinct elements of $V_i$.
\end{enumerate}
\end{property}

\begin{property}[Selection] \label{lem:MCTselection}
Let $e \in E$ be an edge of $G$ and let $T$ be a tree of size $q$ compatible with $\{ C, \tilde C \}$.  
Then, 
$T$ refines $\restr{S_e}{\feuille{T}}$ if and only if at least one endpoint of $e$ is not in $\feuille{T}$.
\end{property}

\paragraph{Correctness of the reduction.}
It is clear that $(q, \calt)$ is computable in polynomial time from $(k, G, V_1, V_2, \ldots, V_k)$.
Moreover, 
both $C$ and $\tilde C$ are binary trees, 
and all internal nodes in $S_e$ have degree two, 
except maybe $\lambda_k^{i, j}$ whose degree is at most $2 \ceil{\log k} + 1$ (see Property~\ref{prop:deg-lambda}).
Hence, the maximum degree $D$ over all trees in $\calt$ is at most $2 \ceil{\log k} + 1$, 
and thus $2^{\left\lfloor D \quotient 2 \right\rfloor} = O(k)$.
Eventually, 
it remains to show that: 
$(k, G, V_1, V_2, \ldots, V_k)$ is a yes-instance of \ppisdeux{}
if and only if 
$(q, \calt)$  is a yes-instance of \past.
\begin{trivlist}
\itemif{}
Assume that there exists a tree $T$ of size $q$ compatible with $\calt$.
Let $I := \feuille{T}$.
By Property~\ref{lem:MCTcontrol}
$I \cap V_i$ is a doubleton for every $i \in \seg{1}{k}$.
By Property~\ref{lem:MCTselection}, $I$ is an independent set of $G$. 
\itemonlyif{}
Conversely, 
assume that there exists an independent set $I$ of $G$ \suchthat{} $I \cap V_i$ is a doubleton for all $i \in \seg{1}{k}$.
For each $i \in \seg{1}{k}$, let $a_i$ and $b_i$ be such that $I \cap V_i = \{ a_i, b_i \}$.
The tree $T := H_k[\lr{a_1, b_1}, \lr{a_2, b_2}, \ldots, \lr{a_k, b_k}]$ is compatible with $\{ C , \tilde C \}$ according to Property~\ref{lem:MCTcontrol}. 
Furthermore, 
$T$ is also compatible with $\left\{ S_e : e \in E \right\}$ according to Property~\ref{lem:MCTselection}.
We have thus exhibited a tree $T$ of size $q$  compatible with $\calt$. 
\end{trivlist}
\begin{remark}
$2^{\left\lfloor D \quotient 2 \right\rfloor} = O(k)$ is enough to obtain Result~$(R2$).
But, 
our construction does not ensure that $2^{\left\lfloor D \quotient 2 \right\rfloor}$ is a function of $k$ only.
Hence, our reduction is not exactly an FPT-reduction yet.
Anyway,
this can be easily repaired. 
Collapse $2 \ceil{\log k} - 1$ consecutive internal edges in $B_1$ to obtain a tree $B_1'$ of maximum degree $2 \ceil{\log k} + 1$ 
and 
add to $\calt$ the tree $C' := H_k[B_1', B_2, \ldots, B_k]$.
\end{remark}
\end{proof}

\bibliographystyle{plain}
\bibliography{CPM06}

\end{document}